%% file: ebay.tex
\documentclass[11pt]{article}

\usepackage{fullpage}
\usepackage{float}
\usepackage{epsfig}
\usepackage{amsmath}
\usepackage{amssymb}
\usepackage{amstext}
\usepackage{amsmath}
\usepackage{xspace}
\usepackage{latexsym}
\usepackage{verbatim}
\usepackage[ruled,vlined]{algorithm2e}

\newtheorem{theorem}{Theorem}[section]

\newtheorem{lemma}[theorem]{Lemma}

\newtheorem{corollary}[theorem]{Corollary}

\newcommand{\qed}{\mbox{\ \ \ }\rule{6pt}{7pt} \bigskip}

\renewcommand{\comment}[1]{}
\newenvironment{proof}{\noindent{\em Proof:}}{\hfill\qed}

\input{macros.tex}

\newcommand{\rev}{\mathcal R}
\newcommand{\cdist}{\mathcal C}
\newcommand{\vdist}{\mathcal V}

\newcommand{\hr}{h}
\newcommand{\exdist}{{E}}

\newcommand{\bprice}{p^B}
\newcommand{\sprice}{p^S}

\newcommand{\Ebp}{Market Intermediation }
\newcommand{\ebp}{market intermediation }

\begin{document}
\title{eBay's Market Intermediation Problem}

\author{Kamal Jain\thanks{eBay Research Labs, San Jose, CA.}
\and Christopher A. Wilkens\thanks{University of California at Berkeley, Berkeley, CA. This work was done while the author was an intern at eBay Research Labs.}
}

\date{}
\maketitle{}

\thispagestyle{empty}

\begin{abstract} We study the optimal mechanism design problem faced by a market intermediary who makes revenue by connecting buyers and sellers. We first show that the optimal intermediation protocol has substantial structure: it is the solution to an algorithmic pricing problem in which seller's costs are replaced with virtual costs, and the sellers' payments need only depend on the buyer's behavior and not the buyer's actual valuation function.

Since the underlying algorithmic pricing problem may be difficult to solve optimally, we study specific models of buyer behavior and give mechanisms with provable approximation guarantees. We show that offering only the single most profitable item for sale guarantees an $\Omega(\frac1{\log n})$ fraction of the optimal revenue when item value distributions are independent and have monotone hazard rates. We also give constant factor approximations when the buyer considers all items at once, $k$ items at once, or items in sequence.

\end{abstract}

\section{Introduction}

The internet has connected people in ways unfathomable a few decades ago. This newfound ease of communication has ushered in a new era in commerce, yet, alone the internet is just a tool --- without intermediaries like Google, eBay, and Amazon to help buyers navigate and connect with sellers, e-commerce would not have flourished. In particular, sites like eBay and Amazon have established themselves as nexuses connecting buyers and sellers in ways that were impossible without the internet. This raises a natural design question: given a group of buyers and sellers, what is the optimal way to intermediate between them? We call this the \ebp problem.

Market intermediation is fundamental to e-commerce. First, it solves an informational problem by connecting buyers and sellers who were not already aware of each other. For example, without a central portal like Google or eBay, a buyer might only shop at the online stores for sellers he already knew about offline. Second, an intermediary reduces transaction costs by making the most appropriate connections between buyers and sellers. Even if buyers knew all the sellers on the internet, there is a small bargaining cost for a buyer and seller to exchange information about supply, demand, and prices. The wide variety of buyers and sellers on the internet would make such costs prohibitively large --- an intermediary reduces costs by aggregating information, allowing buyers and sellers to be matched efficiently.

Our work has roots in the classic economic analysis of bilateral trade and double auctions. In 1983, Myerson and Satterthwaite~\cite{MS83} asked a simple question: will two self-interested people always trade when it is beneficial? The intriguing answer --- there is no equilibrium in which two parties trade whenever it is beneficial without the help of an intermediary --- derives from constraints on truthful payments, and their subsequent analysis of ``Trading with a Broker'' precisely models our scenario with a single buyer and a single seller. Myerson and Satterthwaite's result~\cite{MS83} was the first in a line of research on double auctions (e.g.~\cite{FR93}), where many buyers and sellers place bids on a single commodity. More recent contributions from computer scientists have considered efficient algorithms and approximations (e.g.~\cite{WWW98,DGHK02,DGTZ12}) as well as the effect of a network structure~\cite{BEKT07}. The most relevant work to our own is a very recent paper by Deng et al.~\cite{DGTZ12}, who give approximation algorithms for optimizing revenue in Bayesian settings with multiple buyers, sellers, and items.

The \ebp problem also has roots in the more recent computer science study of algorithmic pricing: given a collection of items and knowledge of buyer behavior, how should the items be priced? This basic computational question was first studied by Rusmevichientong et al.~\cite{RVG06} and Aggarwal et al.~\cite{AFMZ04} in a rank-buying model. Subsequent research considered a wide variety of other models, the most relevant results being Chawla, et al.~\cite{CHK07,CHMS10} who gave approximations for the case of a single unit-demand buyer and Cai and Daskalakis~\cite{CD11} who gave an FPTAS for the same setting.

The \ebp problem combines these two research areas: interaction with the buyer is essentially an algorithmic pricing problem, yet the items for sale are brought by selfish agents as in double auctions.

\paragraph{Results.} We study the \ebp problem in a setting with a single unit-demand buyer and a group of $n$ sellers. First, we show that the optimal solution to this problem has substantial structure --- it is monotone for the sellers (necessary and sufficient for incentive compatibility), buyer prices can be reduced to an algorithmic pricing problem in which the intermediary is itself the seller, and the payments to the sellers need only depend on the costs and likelihood of sale (not the precise value of the buyer).

The \ebp problem thus reduces to a computational one: the intermediary must solve an algorithmic pricing problem to choose the prices the buyer sees and then compute payments to the sellers. Since the sellers are single-parameter agents, their payments are given by a straightforward integral expression~\cite{M81,AT01}, and even if that integral is impossible to compute we can use tricks like Archer et al.~\cite{APTT03} to maintain truthfulness in expectation. On the other hand, the optimal prices for the buyer may be more difficult to compute, so we study a variety of efficient approximation strategies for the intermediary.

The first approximation we consider is the simplest strategy for the intermediary: pick the single best item and sell it. We show that this guarantees an $\Omega(\frac1{\log n})$ fraction of the optimal revenue when the items' values are drawn from independent, monotone hazard rate distributions.

Next, we consider approximations when the buyer picks his favorite item from a set. When the buyer considers all items, we adapt the anonymous virtual reserve price of Chawla et al.~\cite{CHMS10} to give a $\frac12$-approximation to the optimal revenue.\footnote{A result of Deng et al.~\cite{DGTZ12} implies a weaker $\frac14$-approximation in this setting. Their result is technically incomparable to ours --- they study a more general setting when there are multiple buyers who each demand an arbitrary number of goods, but they make substantially stronger assumptions about the representation of probability distributions.} In the more realistic setting where the intermediary can show the buyer at most $k$ items, the intermediary must choose which $k$ items to show and then compute prices. Here, the intermediary wants to approximate the optimal revenue that can be obtained by selling at most $k$ items. We show that an anonymous virtual reserve price can be combined with a greedy algorithm to give a $\frac{e-1}{2e}$-approximation, provided the intermediary can estimate the expected maximum value of a set of items (results of Cai and Daskalakis~\cite{CD11} show that this requirement is nontrivial, as we discuss later).

Finally, we consider a buyer who looks at items in order until he finds one he is willing to buy. We show that the order is less important than the prices: choosing any ordering and computing the optimal prices for that order gives a $\frac12$-approximation to the revenue of the optimal cascade.

\section{The \Ebp Problem}\label{sec:prelim}

The \ebp problem models an intermediary connecting a buyer to a group of $n$ sellers. Each seller $i$ is offering a distinct, indivisible item and faces a privately-known cost of $c_i\sim\cdist_i$ when a sale is made (e.g. the on-demand manufacturing cost). Sellers report their costs when they come to the mechanism. We use $c$ to denote the vector of costs and $\cdist=\cdist_1\times\dots\times\cdist_n$ to denote the distribution of $c$. We assume that the distributions for costs are independent and have monotone hazard rates\footnote{The monotone hazard rate requirement is applied to the distribution $-\cdist_i$ so as to guarantee that the virtual cost $\theta_{\cdist_i}(c_i)$ is monotone.} (see Section~\ref{sec:mhr}).

 A unit-demand buyer with type $t\in T$ (where $T$ is independent of $\cdist$) comes to the intermediary to buy an item. After interacting with the intermediary, the buyer purchases item $i$ with probability $x_i(t,c)$, paying an (expected) price $\bprice_i(t,c)$.

The intermediary otherwise has complete control over the interaction between the buyer and the seller --- its goal is to maximize revenue through a truthful protocol, that is, to pay sellers in such a way that revealing their true costs $c$ is a dominant strategy.

\paragraph{The Independent Values Setting.} Our algorithmic results assume the buyer's valuation is described by $n$ values $v_i$ drawn from independent distributions $\vdist_i$.

The buyer is generally presented with a set of prices $\bprice_i$ and chooses the item that maximizes $v_i-\bprice_i$; however, he may not consider all items available. For example, a buyer on eBay may only consider the first page of search results or may look at items until he finds something he likes. We discuss specific attention models in Sections~\ref{sec:lfa} and~\ref{sec:ca}.

\subsection{Distributions, Hazard Rates, and Virtual Values}\label{sec:mhr}

We represent distributions with density functions $f(z)$ and cumulative density functions $F(z)$. For such a distribution $F(z)$, the {\em hazard rate} is given by
\[\hr(z)=\frac{f(z)}{1-F(z)}\enspace.\]
The following identity directly relates $F(z)$ to the hazard rate:
\begin{equation}F(z)=1-e^{-\int_0^z\hr(x)dx}\enspace.\label{eqn:hrid}\end{equation}
A noteworthy distribution is the exponential distribution $\exdist_\gamma(z)=1-e^{-\gamma z}$, which has a constant hazard rate $\hr_{\exdist_\gamma}(z)=\gamma$.

A distribution $F$ is said to have a {\em monotone hazard rate (MHR)} if $\hr(z)$ is non-decreasing. MHR distributions are convenient because the Myersonian {\em virtual value}~\cite{M81} is monotone. Myerson's virtual value $\phi_F(v)$ represents the marginal revenue obtained by selling to a buyer with value $v\sim F(v)$:
\[\phi_F(v)=v-\frac1{\hr(v)}=v-\frac{1-F(v)}{f(v)}\enspace.\]
For this work, we also define the {\em virtual cost} of a seller as
\[\theta_F(c)=c+\frac{F(c)}{f(c)}\enspace.\]
This represents the marginal cost to the intermediary when a seller with cost $c$ makes a sale.

\subsection{Optimal Single-Item Mechanisms}
We highlight two important classical optimality results. First, Myerson's optimal auction~\cite{M81} says that the optimal way to sell a single item to a single buyer with an MHR distribution is to choose the price $p_M$ satisfying $\phi_\vdist(p_M)=0$. When the seller has an outside option $\rev_{OUT}$ and faces a cost $c$, Myerson's result says that the optimal price $\eta$ satisfies
\[\eta-c-\rev_{OUT}=\frac1{\hr_\vdist(\eta)}\enspace.\]

Myerson and Satterthwaite~\cite{MS83} later considered the case where a selfish player holds the item (instead of the auctioneer). Their result precisely corresponds to the \ebp problem with one seller ($n=1$):
\begin{theorem}[Myerson and Satterwhaite]
In the revenue optimal intermediation with a single seller and MHR distributions, the buyer purchases the item if and only if $\phi_\vdist(v)-\theta_\cdist(c)\geq0$.
\end{theorem}

\section{Revenue Optimization in the General Setting}\label{sec:revopt}

In this section, we show that the revenue optimal intermediation scheme has three nontrivial features:
\begin{enumerate}
\item {\em The revenue optimal scheme is monotone for the sellers.} Following classical results for single-parameter agents~\cite{M81,AT01}, this is a necessary and sufficient criterion for seller incentive compatability.
\item {\em The prices seen by the bidder $\bprice_i$ can be computed as if the intermediary were the seller and faced a cost of $\theta_{\cdist_i}(c_i)$.} This gives a natural reduction from the \ebp problem to an algorithmic pricing problem.
\item {\em The payments to the sellers $\sprice_i$ only depend on the costs and purchase decision of the buyer, not on the exact values $v$.} This implies that the intermediary need not ask the buyer for his actual value.
\end{enumerate}
Algorithm~\ref{alg:general} gives a general mechanism for the \ebp problem incorporating these features.
\begin{algorithm}[h]\label{alg:general}
\SetKwInOut{Input}{input}\SetKwInOut{Output}{output}
\SetKw{KwC}{Compute}
\SetKw{KwSolicit}{Solicit}
\SetKw{KwDisp}{Display}
\SetKw{KwPay}{Pay}
\SetKwBlock{CBlock}{Compute...}{}
\nl \KwSolicit costs $c_i$ from sellers\;
\nl \KwC virtual costs $\theta_{\cdist_i}(c_i)$\;
\nl \KwC a protocol for interacting with the buyer, including what items are shown and at what prices\tcp*[r]{The protocol may be restricted when the buyer's attentionis limited.}
\nl \KwPay the seller whose item was bought $\sprice_i(c)$\tcp*[r]{$\sprice_i$ need not depend on the precise value of the buyer, only on the likelihood of sale.}
\caption{A general mechanism for the \ebp problem.}
\end{algorithm}

We first prove seller monotonicity:
\begin{theorem}
The revenue optimal instantiation of the general algorithm (Algorithm~\ref{alg:general}) is monotone for the sellers, i.e. the probability of selling item $i$ is decreasing in $c_i$.
\end{theorem}
\begin{proof} Suppose $T$ is represented as a probability measure. Define the cost vector $c^{+i}$ as the vector that increases $i$'s cost by $\delta$, i.e.
\[c^{+i}_j=\begin{cases}c_j+\delta&i=j\\c_j&otherwise.\end{cases}\]

Consider the optimal intermediation scheme for a type $t$ and cost vector $c$. Let $x_i(t,c)$ denote the probability that the buyer purchases item $i$ and $\bprice_i(t,c)$ denote the expected price for this scheme as defined in Section~\ref{sec:prelim}.

Since the distribution $T$ is independent of the cost distributions, the buyer's behavior does not depend on $c$. Thus, when the types and costs are $t$ and $c$ respectively, optimality implies that the intermediary makes more money selling according to $x_i(t,c)$ and $\bprice_i(t,c)$ than if it acted as if the cost vector was $c^{+j}$. Thus:
\[\sum_i\int x_i(t,c)(\bprice_i(t,c)-\theta_{\cdist_i}(c_i))dT\geq\sum_i\int x_i(t,c^{+j})(\bprice_i(t,c^{+j})-\theta_{\cdist_i}(c_i))dT\enspace.\]
Likewise, optimality for $c^{+j}$ implies
\[\sum_i\int x_i(t,c^{+j})(\bprice_i(t,c^{+j})-\theta_{\cdist_i}(c_i^{+j}))dT\geq\sum_i\int x_i(t,c)(\bprice_i(t,c)-\theta_{\cdist_i}(c_i^{+j}))dT\enspace.\]
Combining these two and canceling terms gives
\[\int x_j(t,c)(\theta_{\cdist_j}(c_j+\delta)-\theta_{\cdist_j}(c_j))dT\geq\int x_j(t,c)(\theta_{\cdist_j}(c_j+\delta)-\theta_{\cdist_j}(c_j))dT\enspace.\]
Since the virtual costs are assumed to be monotone, this implies
\[\int x_j(t,c)dT\geq\int x_i(t,c)dT\]
as desired.

\end{proof}

Given monotonicity, we can apply classical results to compute the revenue and show that the intermediary's pricing problem can be formulated as an algorithmic pricing problem with virtual costs:
\begin{theorem}
The revenue-optimal intermediation protocol can be described as follows: for each cost vector $c$, solve the underlying algorithmic pricing problem for $p$ assuming the cost of item $i$ is $\theta_{\cdist_i}(c_i)$.
\end{theorem}

\begin{proof}
Let $x_i(t,c)$ be the probability that the buyer purchases item $i$ and let $\pi_t$ be the probability that it has type $t$. From~\cite{M81,AT01} we know that for incentive compatibility the payment of seller $i$ must be
\[\sprice_i(t,c)=-c_ix(t,c)-\int_0^{c_i}x(t,c_{-i}u)du\enspace.\]
We can thus write the expected revenue of the mechanism as
\[\ex[\rev]=\int f_{\cdist}(c)\int\sum_i(\bprice_i(c)x_i(t,c)+\sprice_i(t,c))dTdc\enspace.\]
Rearranging in the standard way gives
\[\ex[\rev]=\int f_{\cdist}(c)\int\sum_ix_i(t,c)(\bprice_i(c)-\theta_{\cdist_i}(c_i))dTdc\enspace.\]
That is, the revenue of the intermediary is precisely as if he faces a cost of $\theta_{\cdist_i}(c_i)$ for item $i$.
\end{proof}

\begin{corollary}\label{cor:algebayapprox}
An $\alpha$ approximation to the underlying algorithmic pricing problem without costs gives an $\alpha$ approximation to the \ebp problem as long as it is monotone from the sellers' perspectives.
\end{corollary}

\begin{proof}
Given the true cost vector $c$ (incentive compatibility is guaranteed by monotonicity), the virtual costs $\theta_{\cdist_i}(c_i)$ are constants. For each item $i$, we solve the algorithmic pricing problem with the distribution $\vdist_i$ shifted by $\theta_{\cdist_i}(c_i)$. This guarantees an $\alpha$ approximation for every $c$ and thus an $\alpha$ approximation to the \ebp problem.
\end{proof}

Finally, we show that the intermediary need not explicitly query the buyer's value:
\begin{theorem}
The revenue optimal instantiation of the general algorithm can be implemented so that seller payments depend only on costs and which item is bought, not the precise value of the buyer.
\end{theorem}
\begin{proof} Suppose truthful payments $\sprice_i(t,c)$ exist. Charge $\sprice_i(c)=\ex_{t\sim T}[\sprice_i(t,c)]$.
\end{proof}

\section{Selling the Best Single Item}\label{sec:singleitem}

A simple strategy for the intermediary is to find the single most profitable item and only show that item to the buyer. This strategy can generally guarantee at most a $\frac1n$ fraction of the optimal revenue; however, we show that this gives a much better approximation when the value of each item is independent and drawn from a MHR distribution $\vdist_i$:
\begin{theorem}\label{thm:singleitem}
When the buyer's valuations $\vdist_i$ are independent and MHR, selling the single most profitable item (based on virtual costs $\theta_{\cdist_i}(c_i)$) in an optimal auction generates at least a $\frac1{e\ln n}$ fraction of the optimal revenue in either the full attention or cascade model.
\end{theorem}

To prove Theorem~\ref{thm:singleitem}, we first to show that the benefit from selling an extra item with distribution $\vdist$ in the cascade model is maximized when $\vdist$ is an exponential distribution. We then bound the approximation for items with exponential distributions.

Suppose a seller wishes to sell a single item to a buyer whose value is drawn from the distribution $\vdist$ (the seller's cost is $0$).  Let $\rev_{M}$ be the revenue generated by selling the item by itself using Myerson's optimal auction, and let $p_M$ be the price at which Myerson's auction sells the item.

Suppose also that the seller has an outside option guaranteeing revenue of $\rev_{OUT}$ if the item does not sell. Let $\rev_\vdist>\rev_M$ be the revenue generated by the optimal auction when the seller considers his outside option, and let $\Delta\rev_\vdist=\rev_\vdist-\rev_{OUT}$ be the optimal incremental revenue, i.e. the extra benefit to the seller on top of his outside option. Then we have the following lemma:
\begin{lemma}\label{lem:expdom}
For a fixed value of $\rev_{M}$, the incremental revenue $\Delta\rev_\vdist$ is maximized when $\vdist$ is an exponential distribution for any value of $\rev_{OUT}$. In particular,
\[\Delta\rev_\vdist\leq\frac1{e\gamma}e^{-\gamma\rev_{OUT}}\]
where $\gamma=\frac1{e\rev_M}$.
\end{lemma}

\begin{proof} First, we identify three distributions:
\begin{itemize}
\item $\vdist$ is the true distribution of $v$ as defined in the problem.
\item $E$ is the exponential distribution that generates the same revenue as $\vdist$ in Myerson's optimal auction. We will see that this is the exponential distribution with $\gamma=\frac1{e\rev_M}$.
\item $\tilde\vdist$ is a hybrid defined through its hazard rate:
\[\hr_{\tilde \vdist}=\begin{cases}\hr_\vdist(v),&v\leq p_M\\
\hr_\vdist(p_M),&otherwise.
\end{cases}\]
(The fact that $\tilde\vdist$ is well-defined follows from identity (\ref{eqn:hrid}).)
\end{itemize}
Our proof shows that $\Delta\rev_\vdist\leq\Delta\rev_{\tilde\vdist}\leq\Delta\rev_E$.

{\em Preliminaries.} For a distribution $F$, let $\eta_F$ denote the optimal selling price if we account for the seller's outside option. Thus, $\eta_F$ must satisfy
\[\eta_F-\rev_{OUT}=\frac1{\hr_F(\eta_F)}\enspace.\]
The total (optimal) revenue $\rev_F$ generated by $F$ is therefore
\[\rev_F=(1-F(\eta_F))\eta_F+F(\eta_F)\rev_{OUT}\enspace.\]
Thus, the incremental revenue of selling at a fixed price $p$ is
\begin{align*}
\rev_F(p)-\rev_{OUT}&=(1-F(p))p+F(p)\rev_{OUT}-\rev_{OUT}\\
&=(1-F(p))(p-\rev_{OUT})\enspace,
\end{align*}
and the optimal incremental revenue $\Delta\rev_F$ is
\[\Delta\rev_F=(1-F(\eta_F))(\eta_F-\rev_{OUT})=\frac{1-F(\eta_F)}{\hr_F(\eta_F)}=\rev_M\frac{\frac{1-F(\eta_F)}{1-F(p_M)}}{\frac{\hr_F(\eta_F)}{\hr_F(p_M)}}\enspace.\]

{\em Step 1: $\Delta\rev_\vdist\leq\Delta\rev_{\tilde\vdist}$.}

Let $z$ satisfy
\[1-F_{\tilde\vdist}(z)=1-F_{\vdist}(\eta_{\vdist})\enspace.\]
Since $\hr_\vdist(v)$ is assumed to be nondecreasing, the identity (\ref{eqn:hrid}) implies
\[1-F_{\tilde\vdist}(v)=e^{-\int_0^v\hr_{\tilde\vdist}(u)du}\geq e^{-\int_0^v\hr_{\vdist}(u)du}=1-F_{\vdist}(v)\enspace,\]
and thus $z\geq\eta_{\vdist}$ because $F$ is nondecreasing. We know that the optimal revenue $\rev_{\tilde\vdist}$ is at least the revenue generated by selling at price $z$, so
\begin{align*}
\Delta\rev_{\tilde\vdist}&\geq(1-F_{\tilde\vdist}(z))(z-\rev_{OUT})\\
&\geq(1-F_{\vdist}(\eta_\vdist))(\eta_\vdist-\rev_{OUT})
=\Delta\rev_\vdist\enspace.
\end{align*}
Thus, the incremental revenue from $\tilde\vdist$ is at least as large as the incremental revenue from $\vdist$.

{\em Step 2: $\Delta\rev_{\tilde\vdist}\leq\Delta\rev_{E}$.}

As shown above, the incremental revenue $\Delta\rev_{\tilde\vdist}$ is
\[\Delta\rev_{\tilde\vdist}=\rev_M\frac{\frac{1-F_{\tilde\vdist}(\eta_{\tilde\vdist})}{1-F_{\tilde\vdist}(p_M)}}{\frac{\hr_{\tilde\vdist}(\eta_{\tilde\vdist})}{\hr_{\tilde\vdist}(p_M)}}=\rev_M\frac{e^{-(\eta_{\tilde\vdist}-p_M)\hr_{\vdist}(p_M)}}{1}\enspace.\]
By the optimality of $\eta_{\tilde\vdist}$ and $p_M$ as well as the definition of $\tilde\vdist$ we know that
\begin{align*}
\eta_{\tilde\vdist}-\rev_{OUT}=\frac1{\hr_{\tilde\vdist}(\eta_{\tilde\vdist})}=\frac1{\hr_\vdist(p_M)}=p_M\enspace,
\end{align*}
and so $\eta_{\tilde\vdist}-p_M=\rev_{OUT}$. Thus,
\[\Delta\rev_{\tilde\vdist}=\rev_Me^{-\rev_{OUT}\hr_{\vdist}(p_M)}\enspace.\]
Similar analysis of the exponential distribution $E$ shows that
\[\Delta\rev_{E}=\rev_Me^{-\rev_{OUT}\gamma}\enspace.\]
Thus, to show $\Delta\rev_{E}\geq\Delta\rev_{\tilde\vdist}$ we need only show $\gamma\leq\hr_{\vdist}(p_M)$.

One can check that the optimal revenue of $E$ in the absence of $\rev_{OUT}$ is
\[\frac{1-F_{E}(\frac1\gamma)}{\hr_{E}(\frac1\gamma)}=\frac1{e\gamma}\enspace.\]
Thus, by definition of $E$ we have chosen $\gamma$ so that $\rev_M=\frac{1}{e\gamma}$. Using the facts that $\rev_M=\frac{1-F_{\vdist}(p_M)}{\hr_\vdist(p_M)}$, $\hr_\vdist$ is nondecreasing, and $p_M=\frac1{\hr_\vdist(p_M)}$ we get
\begin{align*}
\hr_{\vdist}(p_M)&=e\gamma(1-F_\vdist(p_M))\\
&=e\gamma e^{-\int_0^{p_M}\hr_\vdist(v)dv}\\
&\geq e\gamma e^{-p_M\hr_\vdist(p_M)}= e\gamma e^{-1}=\gamma
\end{align*}This gives
\[\Delta\rev_{\tilde\vdist}=\rev_Me^{-\rev_{OUT}\hr_{\vdist}(p_M)}\leq\rev_Me^{-\rev_{OUT}\gamma}=\Delta\rev_{E}\]
as desired.

\end{proof}

We can now prove our approximation theorem:

\begin{proof} (of Theorem~\ref{thm:singleitem}) Following Corollary~\ref{cor:algebayapprox}, it is sufficient to consider the pricing problem of a single seller without costs. Let $\Pi_{OPT}$ be the optimal ordering of items in a cascade. Let $\rev_k$ be the revenue if we only sold the last $k$ items, and let $\Delta\rev_k=\rev_k-\rev_{k-1}$ be the incremental benefit of being able to sell the $k$th item. Then by Lemma~\ref{lem:expdom} we have
\[\rev_k=\rev_{k-1}+\Delta\rev_k\leq\rev_{k-1}+\frac1{e\gamma_k}e^{-\gamma_k\rev_{k-1}}\leq\rev_{k-1}+\frac1{e\gamma_{min}}e^{-\gamma_{min}\rev_{k-1}}\]
where $\frac1{e\gamma_k}$ is the Myersonian revenue of selling item $i$ and $\gamma_{min}$ is the minimum over all $k$. We use induction to prove that $\rev_k\leq\frac{\ln k}{\gamma_{min}}$.

{\em Base case $k=2$:}
\[\rev_k\leq\frac1{e\gamma_{min}}e^{-\frac1e}<\frac{\ln2}{e\gamma_{min}}\]

{\em Inductive step:}
\begin{align*}
\rev_k&\leq \rev_{k-1}+\frac1{e\gamma_{min}}e^{-\gamma_{min}\rev_{k-1}}\\
&=\frac{\ln(k-1)}{\gamma_{min}}+\frac1{e\gamma_{min}}e^{-\gamma_{min}\frac{\ln(k-1)}{\gamma_{min}}}\\
&\leq\frac1{\gamma_{min}}\left(\ln(k-1)+\ln\left(\frac{k}{k-1}\right)\right)\\
&=\frac{\ln k}{\gamma_{min}}
\end{align*}

Thus, we have
\[\rev_n\leq (e\ln n)\max_i\rev_M^i\]
as desired.

Since the optimal revenue from a cascade is an upper bound on the optimal revenue in a full attention model, this implies a $\frac1{e\ln n}$ approximation for the full attention setting as well.
\end{proof}

\section{Approximation with (Limited) Full Attention}\label{sec:lfa}
As noted earlier, buyers often have limited attention --- a buyer might only look at the first page of search results or may scan down a list until he finds an item he is willing to buy. In this section we use anonymous virtual reserve prices to give approximations when the buyer considers a set of items and picks his favorite. We first consider {\em full attention model}, where the buyer is presented with a set of prices $\bprice_i$ and buyer considers all possible items, then we consider the more realistic {\em $k$-limited attention model} where the buyer only looks at $k$ items. In the later case, the intermediary must choose which $k$ items to show the buyer.

\subsection{Full Attention}

Chawla et al.~\cite{CHMS10} show how to compute an anonymous virtual reserve price that guarantees a $\frac12$-approximation to the algorithmic pricing problem --- we show that this gives a $\frac12$-approximation in the full attention model.

The full attention model is effectively a special case of a double auction. In their very recent work, Deng et al.~\cite{DGTZ12} give a $\frac14$-approximation in the more general setting with multiple buyers, each of whom are interested in buying an arbitrary (finite) number of items. This gives a (weaker) $\frac14$-approximation in our setting; however, their work also makes substantially stronger assumptions about the auctioneer's access to the buyer's distributions. As such, our result is substantially stronger for our single-buyer setting.

\begin{theorem}\label{thm:full2apx}
The $\frac12$-approximation of Chawla et al.~\cite{CHMS10} for the algorithmic pricing problem with a single unit demand bidder gives a $\frac12$-approximation to the \ebp problem.
\end{theorem}

\begin{proof}
Following Corollary~\ref{cor:algebayapprox}, we need only show that the approximation of Chawla et al.~\cite{CHMS10} is monotone, i.e. the likelihood of the seller purchasing item $i$ is decreasing in $c_i$. Their algorithm applied to our setting can be described as follows:
\begin{enumerate}
\item Let $z=\max_{i}\phi_{\vdist_i}(v_i)$, and let $\mathcal Z$ be the distribution of $z$ when $v_i\sim\vdist_i$.
\item Let $r$ be the median\footnote{Chawla, Hartline, Malec, and Sivan~\cite{CHMS10} study a more general setting in which one cannot simply take $r$ to be the median; however, the median suffices for our simpler setting~\cite{CHK07,S84}.} of the distribution $\mathcal Z$.
\item Set $\bprice_i=\phi^{-1}(r)$
\end{enumerate}

Let $s_i=v_i-p_i$ be the surplus if the bidder buys item $i$. Note that the bidder's decision depends only on the surpluses and not on the specific values $v_i$ and $p_i$.

Let $\vdist_i^\delta$ denote the distribution $\vdist_i$ shifted right by $\delta$. Let $p$ be the price vector given by computing this approximation with $\vdist_i$, and let $p^\delta$ be the price vector computed from $\vdist_i^\delta$ for all $i$.

We first show that $p_i^\delta=p_i+\delta$ and that, consequently, the probability the buyer buys any given item is the same. By definition of $\phi$, the distribution of $\phi_{\vdist_i}(v_i)$ shifts by $\delta$:
\[\phi_{\vdist_i^\delta}(v_i+\delta)=v_i+\delta-\frac1{\hr_{\vdist_i}(v_i)}=\phi_{\vdist_i}(v_i)+\delta\enspace.\]
It follows that the median of the max distribution must also shift by $\delta$, so $r^\delta=r+\delta$. Now $p_i^\delta$ is chosen to satisfy
\[\phi_{\vdist_i^\delta}(p_i^\delta)=r^\delta=r+\delta\enspace.\]
Notice that
\[\phi_{\vdist_i^\delta}(p_i+\delta)=\phi_{\vdist_i}(p_i)+\delta=r+\delta\enspace,\]
so $p_i^\delta=p_i+\delta$. Thus, both $v_i$ and $p_i$ increased by $\delta$, so the surpluses remain the same and the probability that the buyer chooses $i$ is the same under $\vdist_i$ as it is under $\vdist_i^\delta$.

Next, consider shifting only the valuation of item $i$ and selling at the price vector $(p_{-i},p_i+\delta)$. Again, the surpluses remain the same, so the probability that the buyer chooses $i$ is the same as it was before the shift. However, since the value $r$ increases every time we shift a distribution to the right, for items $j\neq i$ we know that $p_j$ is lower than the price given by the Chawla et al. approximation. Similarly, we know that $p_i+\delta$ is higher than the price at which the approximation would sell item $i$. Thus, in every outcome the surplus $s_i$ is greater under the approximation than it is under $(p_{-i},p_i+\delta)$ while the surpluses $s_j$ are smaller. It follows that the likelihood the buyer chooses item $i$ can only increase, as desired.

\end{proof}

\subsection{$k$-Limited Attention}

In the $k$-limited attention model, the intermediary can only show $k$ items to the user (e.g. the first page of search results). The buyer subsequently chooses the item maximizing $v_i-\bprice_i$. We show that greedily choosing the set of items to show gives a constant-factor approximation.

Implementing the greedy algorithm requires estimating $\ex[\max_{i\in S}v_i]$ for an arbitrary set $S$. Doing so efficiently is nontrivial --- results of Cai and Daskalakis~\cite{CD11} show that the distribution of $\max_{i\in S}v_i$ may not be sufficiently concentrated to estimate $\ex[\max_{i\in S}v_i]$ from samples for arbitrary distributions. We leave the problem of estimating $\ex[\max_{i\in S}v_i]$ as an open question.

Conditioned on the intermediary's ability to estimate $\ex[\max_{i\in S}v_i]$, the greedy algorithm can be combined with the anonymous virtual reserve construction of Chawla et al.~\cite{CHMS10} to get a constant-factor approximation:
\begin{theorem}\label{thm:klimapx}
Assuming the auctioneer can efficiently estimate $\ex_{v_i\sim\vdist_i}[\max_{i\in S}\phi_{\vdist_i}(v_i)]$ for any set $S$, using a greedy algorithm to select which $k$ items to show and then using the approximation of Chawla et al.~\cite{CHMS10} to compute prices gives an $\frac{e-1}{2e}$ approximation in the $k$-limited attention model.
\end{theorem}

\begin{proof} The approximation of Chawla et al.~\cite{CHMS10} used in Theorem~\ref{thm:full2apx} guarantees a $\frac12$-approximation to $\ex_{v_i\sim\vdist_i}[\max_{i}\phi_{\vdist_i}(v_i)]$, which is the Myersonian revenue if the buyer were split into $n$ competing buyers (one for each good). This is itself an upper bound on the optimal revenue in the single-buyer setting.

Thus, if we knew the set $S$ of items (where $|S|=k$) that maximized the expected virtual value $\ex_{v_i\sim\vdist_i}[\max_{i\in S}\phi_{\vdist_i}(v_i)]$, then Theorem~\ref{thm:full2apx} would imply a $\frac12$-approximation here. Instead, we approximate $S$ using the greedy algorithm. As shown in Lemma~\ref{lem:maxrvgreedy} (below), the greedy algorithm gives a set $S$ whose expected virtual value is an $\frac{e-1}{e}$-approximation to the optimal $S$, giving an $\frac{e-1}{2e}$-approximation to the optimal revenue. The greedy algorithm is monotone from the sellers' perspective, proving the theorem.
\end{proof}

It remains to prove that the greedy algorithm gives an $\frac{e-1}{e}$-approximation as claimed in Theorem~\ref{thm:klimapx}. This will follow from submodularity.

Let $X_1,\dots,X_n$ be independent random variables, and let $S\subseteq[n]$ represent a subset of the $X_i$ variables. Then:
\begin{lemma}\label{lem:maxrv}
The function $f(S)=\ex[\max_{i\in S}X_i]$ is a monotone submodular set function.
\end{lemma}

\begin{proof}
The function $f(S)=\max_{i\in S}X_i$ is a monotone submodular function, therefore by linearity of expectation so is $f(S)=\ex[\max_{i\in S}X_i]$.
\end{proof}

The following lemma follows from the standard analysis of the maximum coverage problem:
\begin{lemma}\label{lem:maxrvgreedy}

Suppose we asked to select a set $S$ of size $k$ that maximizes $\ex[\max_{i\in S} X_i]$. The greedy algorithm gives a $(1-\frac1e)$ approximation for this problem.
\end{lemma}

\section{Approximation with Cascading Attention}\label{sec:ca}

In the cascading attention model, the intermediary has a set of slots. The buyer looks at slots according to an order $\Pi$ and buys the first item for which $v_i\geq p_i$. Knowing this, the intermediary decides the matching of items to slots and the prices $\bprice_i$ for each item

\begin{theorem}\label{thm:cascade2}
Given any fixed ordering $\Pi$ of items, the optimal prices for that ordering generate at least $\frac12$ the revenue of the optimal cascade.
\end{theorem}

\begin{proof} %[of Theorem~\ref{thm:cascade2}] 
Define the cumulative revenue $\rev_k$ under ordering $\Pi$ as the revenue if we only sold the last $k$ items, and the incremental revenue $\Delta\rev_k$ is $\Delta\rev_k=\rev_k-\rev_{k-1}$.

First, we observe that the optimal prices for a given ordering $\Pi$ can be computed greedily. Because distributions are independent, the optimal price for item $k$ is simply the optimal single-item auction price given an outside option of $\rev_{k-1}$. Thus, the optimal prices can be computed recursively starting with the last item in $\Pi$. Lemma~\ref{lem:sia-lem} (below) therefore implies that the incremental revenue of an item is decreasing in its cumulative revenue.

Now, let $z$ be the smallest value of $k$ such that $\rev_k>\frac12\rev^{OPT}$ in the optimal order $\Pi^{OPT}$. If the cumulative revenue of this item in $\Pi$ is at least $\frac12\rev^{OPT}$, then we are done.

Otherwise, the cumulative revenue of each item $k\geq z$ in $\Pi^{OPT}$ is smaller in $\Pi$ than it is in the optimal cascade $\Pi^{OPT}$. This implies that the incremental revenue $\Delta\rev_k$ is greater in $\Pi$ than it is in the optimum. However, the total incremental revenue of items $k\geq z$ is at least $\frac12\rev^{OPT}$ by construction, so the total revenue of $\Pi$ must be at least $\frac12\rev^{OPT}$.
\end{proof}

\begin{lemma}\label{lem:sia-lem}
For an optimal single-item auction, the revenue from selling the item $\eta_{\rev_{OUT}}(1-F(\eta_{\rev_{OUT}}))$ is decreasing in both the outside option $\rev_{OUT}$ and the total revenue of the auction $\rev_{OUT}+\eta_{\rev_{OUT}}(1-F(\eta_{\rev_{OUT}}))$.
\end{lemma}
\begin{proof} Optimality of $\eta_{\rev_{OUT}}$ has a few relevant consequences. First, the optimal selling price $\eta_{\rev_{OUT}}$ is increasing in $\rev_{OUT}$. Second, combined with monotonicity of $\eta_{\rev_{OUT}}$, optimality of $\eta_{\rev_{OUT}}$ implies that the revenue from selling the item $\eta_{\rev_{OUT}}(1-F(\eta_{\rev_{OUT}}))$ is also decreasing in $\rev_{OUT}$. Finally, the total revenue is increasing in $\rev_{OUT}$, so the revenue of selling the item $\eta_{\rev_{OUT}}(1-F(\eta_{\rev_{OUT}}))$ is decreasing in the total revenue.
\end{proof}

\section{Conclusion and Open Questions}

The \ebp problem models a fundamental situation faced by companies like eBay, Amazon, and Google. Much work in economics and algorithmic game theory is relevant, but little directly addresses the problems these companies face. Our results show that the intermediary's problem has significant structure and offers strategies with provable performance guarantees; however, it still leaves many open questions:
\begin{enumerate}
\item Which existing algorithmic pricing solutions are monotone from the sellers' perspective?
\item What other value and attention models permit good approximations?
\item Are there ways the intermediary can leverage mechanisms that do not require direct revelation by the sellers, such as eBay's commission model?
\item How does the model change with multiple bidders?
\item In practice, what is the best intermediation protocol for a company like eBay?
\end{enumerate}

\section{Thanks}

We would like to thank Costis Daskalakis, Darrell Hoy, Christos Papadimitriou, George Pierrakos, Balasubramanian Sivan, and Steve Tadelis for many helpful discussions.

\bibliographystyle{plain}
\bibliography{ebay}

\end{document}

%% file: macros.tex
\newcommand{\hidecite}[1]{}

\newcommand{\ex}{\mathbf{E}}